\documentclass[10pt, twocolumn, comsoc]{IEEEtran}

\usepackage{graphicx,epsfig}
\usepackage[noadjust]{cite}
\usepackage{mcite}
\usepackage{amsfonts,helvet}
\usepackage{fancyhdr}
\usepackage{threeparttable}
\usepackage{epsf,epsfig}
\usepackage{amsthm}
\usepackage{amsmath}
\usepackage{boldline}
\usepackage{booktabs}
\usepackage{amssymb}
\usepackage{lipsum}
\usepackage[caption=false,font=footnotesize]{subfig}
\usepackage{textcomp}
\usepackage{dsfont}

\usepackage{ragged2e} 
\usepackage{multirow} 
\usepackage{makecell} 

\usepackage[colorlinks=false, linkcolor=blue]{hyperref}

\usepackage{dsfont}
\usepackage{color}
\usepackage{array}
\usepackage{algpseudocode}
\usepackage{algcompatible}
\usepackage{enumerate}
\usepackage{gensymb}
\usepackage{cancel}

\usepackage[linesnumbered,ruled,noend]{algorithm2e}

\usepackage{fancyhdr}
\usepackage{graphicx,epsfig}
\usepackage{wrapfig}
\usepackage{ragged2e}
\usepackage{algpseudocode}
\usepackage{algcompatible}
\usepackage{threeparttable}
\usepackage{booktabs} 
\usepackage{tabularx}

\newtheorem{theorem}{Theorem}
\newtheorem{corollary}{Corollary}

\newtheorem{lemma}{Lemma}

\newtheorem{remark}{Remark}

\usepackage{eucal}

\begin{document}

\title{Load Balancing in Multi-Shell LEO Satellite Networks with Successive Interference Cancellation}


\author{Seyong~Kim, Jeonghun~Park, and Jeffrey G. Andrews

\thanks{This work was supported in part by the Institute of Information \& Communications Technology Planning \& Evaluation (IITP) grant funded by the Korea government (MSIT) under 6G Cloud Research and Education Open Hub (IITP-2025-RS-2024-00428780), and in part by the Samsung Research Funding \& Incubation Center of Samsung Electronics under Project Number SRFC-IT2402-06. S. Kim and J. Park are with the School of Electrical and Electronic Engineering, Yonsei University, Seoul 03722, South Korea (e-mail: {\texttt{sykim@yonsei.ac.kr; jhpark@yonsei.ac.kr}}). J. G. Andrews is with the 6G@UT Research Center, Wireless Networking and Communications Group, University of Texas at Austin (e-mail:{\texttt{jandrews@ece.utexas.edu}}).}}


\maketitle

\begin{abstract}
Multi-shell low Earth orbit (LEO) networks can increase service opportunities, but altitude-dependent propagation can concentrate traffic on lower shells and create strong inter-shell interference under full frequency reuse.
This paper develops a mathematical framework for load balancing in multi-shell LEO satellite networks.
Satellites on each shell form an independent spherical Poisson point process (SPPP), and the typical user associates with one of the per-shell serving satellites through a shell-dependent biased received-power rule, with receiver-side successive interference cancellation (SIC) under full frequency reuse.
Shell-wise association probabilities, conditioned serving-distance distributions, and the rate coverage probability under shell-dependent traffic loads are derived and validated by simulation.
The results show that shell-dependent biasing alleviates lower-shell traffic concentration and improves rate coverage, while receiver-side SIC mitigates the dominant lower-shell interference experienced by users associated with upper shells.
Load balancing provides its largest rate-coverage gain in traffic hotspots, while SIC becomes more valuable as receive-side isolation weakens.
With a fixed satellite budget, distributing satellites across multiple shells can further improve hotspot rate coverage by adding shell-wise serving opportunities.
\end{abstract}


\section{Introduction}
Dense low Earth orbit (LEO) mega-constellations such as Starlink deploy large numbers of satellites across multiple low-altitude orbital shells. 
Such satellite networks must provide wide-area connectivity while sustaining reliable per-user rates under highly nonuniform traffic demand \cite{samy:commmag:2022,chen:surv:2026,di:wcom:2019}.
This requirement becomes especially important in traffic hotspots, such as densely populated urban areas, where many users compete for the limited time-frequency resources of the satellites serving the same area.
Multi-shell deployments are motivated not only by coverage expansion but also by capacity scaling \cite{wang:vtmag:2021}. For example, recent LEO constellation filings consider multiple shells with different altitudes and inclinations, where lower-altitude shells can form smaller ground footprints and support denser spatial reuse \cite{fcc:spacex:2020, fcc:kuiper:2021}. 


Motivated by such multi-shell deployments, this paper considers a setting in which satellites from different shells may illuminate the same ground footprint over the same time--frequency resource \cite{fcc2beam:spacex:2026}.
This architecture creates two coupled challenges.
First, lower-altitude shells generally provide shorter slant distances and lower path loss, so received-power-based association preferentially selects them and can concentrate traffic on these shells.
Second, associating users with higher shells can reduce lower-shell traffic concentration, but such users generally experience larger propagation loss and remain exposed to strong lower-shell serving signals under full frequency reuse.
To address these challenges, shell-dependent association biasing redistributes traffic toward higher shells to relieve lower-shell traffic concentration, while receiver-side successive interference cancellation (SIC) cancels dominant lower-shell serving signals to mitigate the interference penalty of upper-shell association.

The objective of this paper is to develop a mathematical framework for quantifying the resulting load-balancing and interference-mitigation effects in multi-shell LEO satellite networks.
Using stochastic geometry, we characterize how shell association and traffic load interact with inter-shell interference to determine rate coverage under different constellation and traffic conditions.



\subsection{Related Work}
Load balancing in multi-layered satellite networks has been studied from several perspectives \cite{li:aina:2018,kawamoto:tvt:2013,nishiyama:tvt:2013,wang:vtc:2015}, with broader surveys provided in \cite{shang:surv:2026,Kodheli:surv:2021}.
For example, \cite{nishiyama:tvt:2013} distributes traffic across LEO and medium Earth orbit layers using capacity and congestion information, while \cite{wang:vtc:2015} develops load-aware routing based on lower-layer load information.
Related control-plane work jointly optimizes controller placement and switch-to-controller assignment according to controller load and satellite mobility \cite{chen:tsc:2024}.
These approaches redistribute traffic across satellite layers, but mainly address routing or control-plane load balancing rather than user association, traffic-dependent resource-sharing, and interference mitigation in co-channel multi-shell LEO access.

Association-stage load balancing introduces an inherent tradeoff between traffic redistribution and interference.
In terrestrial heterogeneous networks, association biasing or cell range expansion redistributes traffic across tiers but can expose offloaded users to strong cross-tier interference \cite{singh:twc:13,jo:twc:2012,Damnjanovic:wcom:2011}.
Receiver-side SIC has been studied to mitigate this interference and preserve the benefits of traffic offloading \cite{weber:tit:2007,zhang:tit:2014,wildemeersch:tcom:2014}.
In particular, \cite{wildemeersch:tcom:2014} jointly considers SIC, minimum-load association, and range expansion.
Related receiver-side SIC techniques have also been considered in multi-layer and LEO satellite systems \cite{lu:twc2025,li:tvt:2026}.
In multi-shell LEO networks, the same tradeoff is shaped by the shell geometry: users offloaded to upper shells experience larger serving distances, while lower-shell serving signals can become dominant interferers and SIC targets.

Stochastic geometry (SG) provides a mathematical framework for analyzing these association and interference effects.
SG has been widely applied to LEO satellite networks to quantify the effects of satellite density, altitude, and propagation on coverage \cite{okati:tcom:20,park:twc:22,choi:tcom:2025}.
For example, \cite{park:twc:22} showed that stochastic satellite models can closely reproduce coverage predictions based on actual Starlink locations, while related works have considered multi-layer constellations with inter-layer interference \cite{okati:commlett:2023,ma:pimrc35:2024}.
In addition, \cite{choi:twc:2024} analyzes association, inter-constellation interference, and SINR coverage for multiple coexisting LEO constellations over different orbits and altitudes.
These studies provide tractable association and coverage analyses, but generally do not couple shell association with traffic-dependent resource-sharing or load-aware biasing.

To the best of our knowledge, shell-dependent association biasing jointly with receiver-side SIC has not been characterized for co-channel multi-shell LEO networks.
In this setting, association, traffic load, and interference are coupled: associating a user with an upper shell changes both its serving-link geometry and shell-wise resource-sharing, while lower-shell serving signals can become dominant interferers and potential SIC targets.
This motivates a rate coverage framework that jointly accounts for shell-dependent association biasing, traffic load, inter-shell interference, and receiver-side SIC.

\subsection{Contributions}

The main contributions of this paper are summarized as follows.

\textbf{Mathematical multi-shell SG framework}: We develop a novel SG-based framework for an \(N\)-shell LEO network, with each shell modeled by an independent spherical Poisson point process (SPPP). The model assumes at most one serving satellite per shell, with up to \(N\) shell-wise serving satellites illuminating the same footprint under full frequency reuse, while the user associates with one according to a biased received-power rule. Non-associated serving satellites generate inter-shell interference, whereas non-serving satellites contribute side-lobe interference. Under this model, we derive shell-wise association probabilities, association-conditioned serving-distance distributions, and rate coverage under shell-dependent loads and receiver-side SIC.

\textbf{Association-stage load balancing with SIC-enabled full reuse}: We introduce shell-dependent association biasing to redistribute traffic from preferentially selected lower shells toward upper shells. The bias trades the resource-sharing gain from offloading against the larger propagation loss of upper-shell service. Under full frequency reuse, receiver-side SIC cancels strong lower-shell serving signals for users associated with upper shells, mitigating the interference penalty of upper-shell offloading without requiring cross-shell transmission coordination or power allocation.

\textbf{Rate coverage verification and design insights}: Monte Carlo simulations validate the analytical results and show that multi-shell load balancing with full frequency reuse is most beneficial in traffic-hotspot regimes.
For the evaluated configurations, denser constellations require stronger upper-shell biasing.
SIC provides larger gains as receive-side spatial isolation weakens.
With a fixed satellite budget, multiple shells create additional shell-wise serving opportunities over the same footprint and can improve hotspot rate coverage compared with concentrating the satellites in a single shell.

The main symbols used throughout this paper are summarized in Table~\ref{tab:notation}.
Section~\ref{sec:system model} presents the multi-shell network model. Section~\ref{sec:shell-wise association prob} analyzes shell association under biasing. Section~\ref{sec:rate coverage} derives the rate coverage with shell-dependent loads and receiver-side SIC. Section~\ref{sec:numerical results} validates the analysis and presents numerical results. Section~\ref{sec:conclusion} concludes the paper.


\begin{table*}[t!]
\caption{Summary of notation.}
\label{tab:notation}
\centering
\footnotesize
\begin{tabularx}{\textwidth}{@{}l X @{\hspace{0.8em}\vrule width 0.4pt\hspace{0.8em}} l X@{}}
\toprule
Symbol & Description & Symbol & Description \\
\midrule
\(\mathcal A_n(r),\ \Lambda_n(r)\) & Visible region within \(r\); intensity measure & \(I_m^{\rm serv},\ I_m^{\rm non},\ \mathcal I_n^{(i)}\) & Serving, non-serving, and residual interference \\
\(\Phi_n^{\rm serv},\ \Phi_n^{\rm non}\) & Serving and non-serving satellite sets & \(\mathcal F_n^{(k)},\ Z_{n,i}\) & Interference when decoding shell \(k\); \(\sigma^2+\mathcal I_n^{(i)}\) \\
\(R_n,\ \mathbf R,\ \mathbf R_{-n}\) & Serving distance; its vector; vector excluding \(R_n\) & \(D_n^{(i)},\ \mathcal D_n^{(i)}\) & Desired-decoding and incremental desired-decoding events \\
\(\mathds{1}_{\{\cdot\}},\ \Omega,\ (\cdot)^c\) & Indicator; sure event; complement & \(C_n^{(k)},\ \mathcal C_n^{(i)}\) & Shell-\(k\) cancellation and cumulative cancellation events \\
\(\rho_{m,n}(r),\ p_{m,n}^{\rm void}(r)\) & Admissible-distance limit; void probability & \(\mathsf S_n,\ Y_k,\ \mathcal S_n(\gamma)\) & Desired power; shell-\(k\) serving power; success event \\
\(\mathcal E_n,\ E_n,\ p_n\) & Existence and association events; \(\mathbb P(E_n)\) & \(M_m^{\rm serv},\ M_m^{\rm non}\) & Serving- and non-serving-link Laplace factors \\
\(K_u,\ \delta,\ L_n\) & Candidate users; concurrent-demand fraction; load & \(\Psi_m^{\rm non},\ \mathcal L_X\) & Aggregate non-serving interference transform; Laplace transform \\
\(\mathcal R_n^{(i)},\ \tau_n\) & Stage-\(i\) rate; required SINR threshold & \(q_{n,i},\ \mathcal P_{{\rm cov},n}(\gamma),\ p_{\rm cov}(\gamma)\) & Stage-\(i\) rate coverage contribution; shell-conditioned and overall rate coverage \\
\bottomrule
\end{tabularx}
\end{table*}

\begin{figure}
    \centering
    \includegraphics[width=1\linewidth]{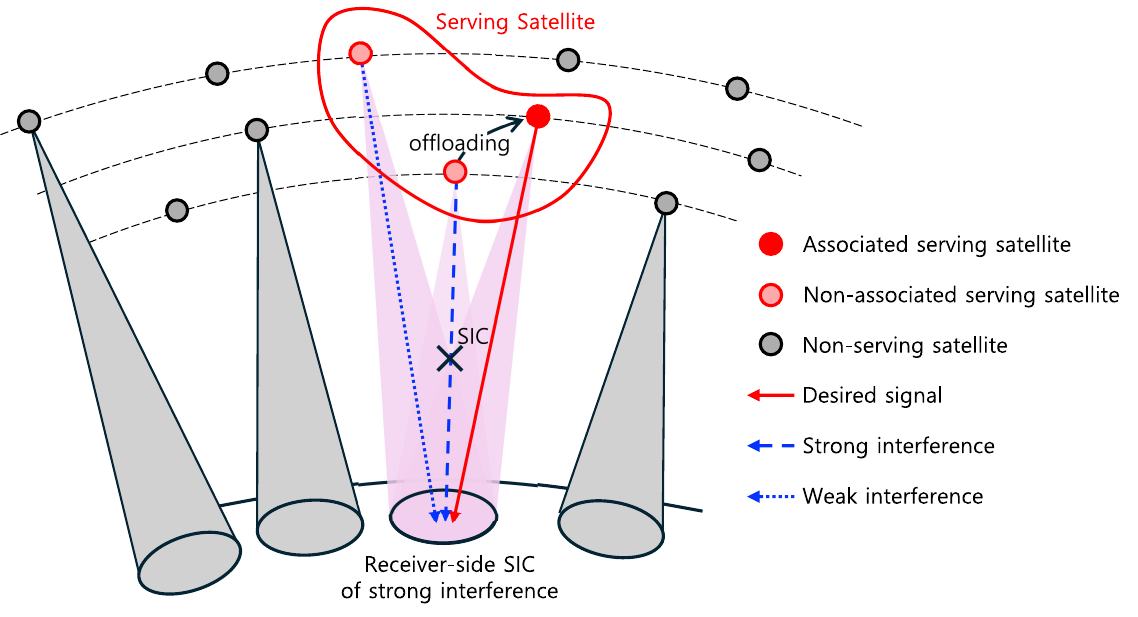}
        \caption{Illustration of the considered three-shell downlink satellite network.}
    \label{fig:system_model}
\end{figure}

\section{System Model}\label{sec:system model}

We consider a downlink satellite network with \(N\) concentric orbital shells indexed by \(n\in\mathcal N=\{1,\ldots,N\}\), ordered by altitude as \(h_1<h_2<\cdots<h_N\). The considered network geometry is illustrated for \(N=3\) in Fig.~\ref{fig:system_model}.
Let \(R_\oplus\) denote the Earth radius and \(h_n\) the altitude of shell \(n\). The radius of shell \(n\) is \(a_n=R_\oplus+h_n\), and its surface is \(\mathbb S_n=\{\mathbf x\in\mathbb R^3:\|\mathbf x\|=a_n\}\).

\subsection{Shell and Association Model}
Satellites on shell \(n\) are modeled as a homogeneous SPPP \(\bar\Phi_n=\{\mathbf d_{1,n},\ldots,\mathbf d_{|\bar\Phi_n|,n}\}\) with density \(\lambda_n\). The SPPPs \(\{\bar\Phi_n\}_{n=1}^N\) are mutually independent, and \(|\bar\Phi_n|\sim{\rm Poisson}(4\pi\lambda_n a_n^2)\). By rotational invariance, the typical user is placed at \(\mathbf u=(0,0,R_\oplus)\).

Let \(\theta_{\min}\) denote the minimum elevation angle. The maximum visible distance from the typical user to shell \(n\) is
\begin{equation}
r_{\max,n}=\sqrt{a_n^2-R_\oplus^2\cos^2\theta_{\min}}-R_\oplus\sin\theta_{\min}. \nonumber
\end{equation}
The visible region on shell \(n\) within distance \(r\) is \(\mathcal A_n(r)=\{\mathbf x\in\mathbb S_n:\|\mathbf x-\mathbf u\|\le r\}\), and \(\mathcal A_n=\mathcal A_n(r_{\max,n})\). The visible satellite process is \(\Phi_n=\bar\Phi_n\cap\mathcal A_n\).

For the considered time-frequency resource block, the visible satellites in shell \(n\) are partitioned into serving and non-serving sets, denoted by \(\Phi_n^{\rm serv}\) and \(\Phi_n^{\rm non}\), respectively. At most one satellite per shell serves the typical footprint. When \(\Phi_n\neq\varnothing\), the serving satellite for $\mathbf{u}$ in shell $n$ is selected as the nearest visible satellite:
\begin{equation}
\mathbf d_{0,n}=\arg\min_{\mathbf d_{j,n}\in\Phi_n}\|\mathbf d_{j,n}-\mathbf u\|. \nonumber
\end{equation}
Accordingly, \(\Phi_n^{\rm serv}=\{\mathbf d_{0,n}\}\), \(\Phi_n^{\rm non}=\Phi_n\setminus\{\mathbf d_{0,n}\}\), and \(R_n=\|\mathbf d_{0,n}-\mathbf u\|\). If \(\Phi_n=\varnothing\), then \(\Phi_n^{\rm serv}=\Phi_n^{\rm non}=\varnothing\).

The typical user can associate only with serving satellites, since non-serving satellites are not beam-aligned with the typical footprint. The associated shell is determined by
\begin{equation}\label{eq:association_rule_serving}
n^*=\arg\max_{n:\Phi_n^{\rm serv}\neq\varnothing} P_nB_nR_n^{-2},
\end{equation}
where \(P_n\) and \(B_n\) denote the transmit power and association bias of shell \(n\), respectively. The bias \(B_n\) controls the shell-wise association probability and enables traffic offloading across shells. If \(P_nB_n\) is identical for all shells, \eqref{eq:association_rule_serving} reduces to nearest-serving-satellite association.

The area of \(\mathcal A_n(r)\) is \(A_n(r)=\pi a_n(r^2-h_n^2)/R_\oplus\) for \(h_n\le r\le r_{\max,n}\), and \(A_n=A_n(r_{\max,n})\). The distance-dependent intensity measure of shell \(n\) is
\begin{equation}
\Lambda_n(r)=
\begin{cases}
0, & r<h_n,\\
\lambda_n A_n(r), & h_n\le r\le r_{\max,n},\\
\lambda_n A_n, & r>r_{\max,n}.
\end{cases} 
\end{equation} 
Its derivative is \(\Lambda_n'(r)=2\pi\lambda_n a_n r\mathds{1}_{\{h_n<r<r_{\max,n}\}}/R_\oplus\).
Let \(\mathcal E_n=\{\Phi_n\neq\varnothing\}\) denote the event that shell \(n\) contains at least one visible satellite, and let \(\mathcal E=\bigcup_{n=1}^N\mathcal E_n\) denote the event that at least one serving candidate is available across all shells.
Association is possible only when \(\mathcal E\) occurs.
By independence of the shell-wise SPPPs and their void probabilities,
\begin{equation}\label{eq:exist p}
p_{\mathcal E}
=\mathbb P(\mathcal E)
=1-\prod_{n=1}^N e^{-\lambda_n A_n}
=1-\exp\left(-\sum_{n=1}^N\lambda_n A_n\right).
\end{equation}

\subsection{Channel and Antenna Models}

We adopt Nakagami-$m_0$ fading because it provides a flexible yet analytically tractable model for a range of fading conditions.
Under Nakagami-\(m_0\) fading, the power gain \(H_{j,n}\sim\mathrm{Gamma}(m_0,1/m_0)\) has unit mean. For integer \(m_0\ge1\) and \(x\ge0\), its complementary cumulative distribution function (CCDF) is
\begin{equation}\label{eq:ccdf gamma}
\mathbb{P}(H>x)=e^{-m_0x}\sum_{q=0}^{m_0-1}\frac{(m_0x)^q}{q!}.
\end{equation}
The fading gains are mutually independent across satellites, links, and shells, and are independent of the point processes, serving-satellite selection, association events, and antenna-gain variables.

Practical satellite receive antennas exhibit a pronounced gain separation between the boresight and off-axis regions, with the off-axis gain approaching a nearly constant level at sufficiently large angular separations \cite{etsi:en303981:22}.
To capture this dominant main- and off-axis gain separation while retaining analytical tractability, we approximate the receive antenna pattern using the two-state sectored model in \cite{bai:twc:2015,hunter:twc:2008}.
The associated serving link has aligned transmit and receive beams and therefore experiences the main-lobe gain product \(G_{\rm t}G_{\rm r}\).
A non-associated serving satellite in shell \(m\neq n\) is also assumed to illuminate the typical user within its transmit main lobe, and hence contributes the transmit gain \(G_{\rm t}\).
Since the user receive beam is aligned with the associated satellite, only the receive-side gain of the non-associated serving link depends on the angular separation between the two satellite directions.
Its effective gain is therefore modeled as \(G_{\rm t}X_m\), where
\begin{equation}\label{eq:antenna gain}
X_m =
\begin{cases}
G_{\rm r}, & \text{with probability } p_{\rm r},\\
g_{\rm r}, & \text{with probability } 1-p_{\rm r}.
\end{cases}
\end{equation}
The variables \(\{X_m\}_{m=1}^N\) are mutually independent and independent of the point processes and fading gains. The probability \(p_{\rm r}\) captures the likelihood that a non-associated serving signal is received through the main lobe of the user antenna. This two-state model can be further generalized to a multi-state discrete model by partitioning the receive antenna pattern into finer angular regions, thereby providing a more refined approximation of practical off-axis gain patterns \cite{etsi:en303981:22}.

If \(\Phi_m^{\rm serv}\neq\varnothing\), the serving interference from shell \(m\neq n\) is
\begin{equation}
I_m^{\rm serv}
=
\alpha P_mG_{\rm t}X_mH_{0,m}R_m^{-2}.
\end{equation}
Otherwise, \(I_m^{\rm serv}=0\).
Let \(\alpha\triangleq\left(\frac{c}{4\pi f_c}\right)^2\) denote the free-space power-gain coefficient, where \(c\) is the speed of light and \(f_c\) is the carrier frequency.

Unlike a non-associated serving satellite, a non-serving satellite is not beam-aligned with the typical footprint and therefore contributes through the transmit side lobe \(g_{\rm t}\).
Since the user receive beam is aligned with the associated serving satellite, its leakage is further modeled through the receive side lobe \(g_{\rm r}\), yielding the fixed gain product \(g_{\rm t}g_{\rm r}\).
The aggregate non-serving interference from shell \(m\) is
\begin{equation}
I_m^{\rm non}
=
\sum_{\mathbf d_{j,m}\in\Phi_m^{\rm non}}
\alpha P_mg_{\rm t}g_{\rm r}H_{j,m}
\|\mathbf d_{j,m}-\mathbf u\|^{-2}.
\end{equation}
This fixed side-lobe model captures the aggregate leakage from satellites whose beams are not directed toward the typical footprint.

When the typical user is associated with shell \(n\), shell-wise SIC cancels only lower-shell serving signals, while all non-serving signals remain as side-lobe interference. After serving signals from shells \(1,\ldots,i\) have been removed, the residual interference is
\begin{equation}\label{eq:residual_interference_serv_non}
\mathcal{I}_n^{(i)}
=
\sum_{\substack{m=i+1\\m\neq n}}^{N}
I_m^{\rm serv}
+
\sum_{m=1}^{N}
I_m^{\rm non},
\end{equation}
where \(i=0,\ldots,n-1\). Thus, \(\mathcal I_n^{(0)}\) is the interference before SIC, and \(\mathcal I_n^{(n-1)}\) is the residual interference after all lower-shell serving signals have been canceled or skipped. This serving/non-serving decomposition captures the two effects of satellite density: larger \(\lambda_n\) shortens the serving distance but increases aggregate side-lobe interference.

\section{Shell-wise Association Probability}\label{sec:shell-wise association prob}

We now characterize the probability that the typical user is associated with shell \(n\). Conditioned on the existence event \(\mathcal E_n\), the serving satellite is the nearest visible satellite in shell \(n\). Hence, for \(h_n<r<r_{\max,n}\),
\begin{equation}\label{eq:Rn_cond_pdf}
f_{R_n\mid \mathcal E_n}(r)
=
\frac{\Lambda_n'(r)e^{-\Lambda_n(r)}}{p_{\mathcal E_n}},
\end{equation}
where \(p_{\mathcal E_n}\triangleq\mathbb P(\mathcal E_n)=1-\exp(-\lambda_nA_n)\) denotes the shell-\(n\) existence probability.
Let \(T_{mn}\) denote the biased exclusion-distance scaling factor from shell \(n\) to shell \(m\):
\begin{equation}\label{eq:Tmn}
T_{mn}
=
\sqrt{\frac{P_mB_m}{P_nB_n}}.
\end{equation}
Conditioned on \(\mathcal E_n\) and \(R_n=r\), shell \(m\) beats shell \(n\) if \(P_mB_mR_m^{-2}>P_nB_nr^{-2}\), or equivalently \(R_m<T_{mn}r\). Since the serving satellite in shell \(m\) is the nearest visible satellite, shell \(m\) does not beat shell \(n\) when \(\Phi_m\cap\mathcal A_m(T_{mn}r)=\varnothing\). This event also includes the case \(\Phi_m=\varnothing\), so \(R_m\) need not be defined when shell \(m\) has no visible satellite.

Based on \eqref{eq:association_rule_serving}, the shell-\(n\) association event is
\begin{equation}
E_n
=
\mathcal E_n
\cap
\bigcap_{m\neq n}
\left\{
\Phi_m\cap\mathcal A_m(T_{mn}R_n)=\varnothing
\right\}. 
\end{equation}
Here, \(E_n\) includes the existence of a serving satellite in shell \(n\), and ties occur with probability zero. Let \(p_n=\mathbb P(E_n)\).

\begin{theorem}[Shell-wise association probability]
\label{thm:shell_wise_association_probability}
The association probability of shell \(n\) is
\begin{equation}\label{eq:shell_wise_association_probability}
p_n
=
\int_{h_n}^{r_{\max,n}}
\Lambda_n'(r)e^{-\Lambda_n(r)}
\exp\!\left[
-\sum_{m\neq n}\Lambda_m(T_{mn}r)
\right]dr.
\end{equation}
\end{theorem}

\begin{proof}
Conditioned on \(\mathcal E_n\) and \(R_n=r\), shell \(m\neq n\) does not yield a larger biased received-power metric than shell \(n\) if there is no visible satellite of shell \(m\) in \(\mathcal A_m(T_{mn}r)\). By the void probability of the SPPP,
\begin{equation}
\mathbb P\!\left(\Phi_m\cap\mathcal A_m(T_{mn}r)=\varnothing\right)
=
\exp[-\Lambda_m(T_{mn}r)]. \nonumber
\end{equation}
Since the SPPPs are independent across shells, the conditional probability that shell \(n\) is preferred over all other shells is
\begin{equation}
\mathbb P(E_n\mid \mathcal E_n,R_n=r)
=
\prod_{m\neq n}\exp[-\Lambda_m(T_{mn}r)]. \nonumber
\end{equation}
Averaging over \(R_n\) conditioned on \(\mathcal E_n\) gives
\begin{equation}
p_n
=
p_{\mathcal E_n}
\int_{h_n}^{r_{\max,n}}
f_{R_n\mid \mathcal E_n}(r)
\prod_{m\neq n}\exp[-\Lambda_m(T_{mn}r)]\,dr. \nonumber
\end{equation}
Substituting \eqref{eq:Rn_cond_pdf} yields \eqref{eq:shell_wise_association_probability}.
\end{proof}

Association is thus governed by bias-scaled exclusion regions: shell \(n\) wins if and only if every other shell is void within \(\mathcal A_m(T_{mn}r)\). Since the bias enters only through the ratio \(P_mB_m/P_nB_n\), fixing \(B_1\) is without loss of generality, and increasing \(B_n\) shrinks the competing exclusion regions and hence increases \(p_n\) monotonically. Also, \(\sum_{n}p_n=p_{\mathcal E}<1\), so \(p_{\mathcal E}\) upper bounds the rate coverage attainable by any bias vector.

\begin{lemma}[Serving-distance distribution conditioned on association]
\label{lem:serving_distance_conditioned_on_association}
For \(h_n<r<r_{\max,n}\), the probability density function (PDF) of the serving distance conditioned on association with shell \(n\) is
\begin{equation}\label{eq:Rn_cond_En_pdf}
f_{R_n\mid E_n}(r)
=
\frac{
\Lambda_n'(r)e^{-\Lambda_n(r)}
\exp\!\left[-\sum_{m\neq n}\Lambda_m(T_{mn}r)\right]
}{
p_n
}.
\end{equation}
\end{lemma}

\begin{proof}
By Bayes' rule, \(f_{R_n\mid E_n}(r)=f_{R_n,E_n}(r)/p_n\), where \(f_{R_n,E_n}(r)\) is the joint density of \(R_n\) and \(E_n\). Since \(E_n\subseteq\mathcal E_n\),
\begin{equation}
f_{R_n,E_n}(r)
=
\mathbb P(E_n\mid \mathcal E_n,R_n=r)
f_{R_n,\mathcal E_n}(r). 
\end{equation}
The first factor is the conditional association probability derived in the proof of Theorem~\ref{thm:shell_wise_association_probability}. The second factor is \(f_{R_n,\mathcal E_n}(r)=p_{\mathcal E_n}f_{R_n\mid\mathcal E_n}(r)=\Lambda_n'(r)e^{-\Lambda_n(r)}\). Combining these terms gives \eqref{eq:Rn_cond_En_pdf}.
\end{proof}

Relative to \eqref{eq:Rn_cond_pdf}, conditioning on association introduces a nonincreasing factor in \(r\), reflecting that shell \(n\) is more likely to be selected when its serving satellite is close. Increasing \(B_n\) reduces \(T_{mn}\), shrinking the distance range over which shell \(m\) can provide a larger biased received-power metric than shell \(n\), and thereby allowing shell \(n\) to be selected at larger serving distances. This conditional density provides the averaging law used throughout Section~\ref{sec:rate coverage}.

\section{Rate Coverage}\label{sec:rate coverage}
This section characterizes rate coverage, defined as the probability that the typical user can support a rate threshold \(\gamma\) through direct decoding or the subsequent SIC procedure.
The probability is unconditional over the satellite point processes, so the typical user is in outage when no visible serving satellite exists.
Since the associated shell determines both the resource-sharing load and the SIC order, the analysis first derives the shell-conditioned rate coverage and then averages it over the shell-wise association probabilities.


\subsection{Resource Allocation}

In the proposed serving/non-serving model, at most one satellite per shell is beam-aligned with the typical footprint over the considered time-frequency resource block.
Accordingly, each occupied shell provides one shell-specific serving link, and the users associated with shell \(n\) share the resources of its serving satellite.

Let \( K_u\) denote the number of candidate users in the typical footprint, excluding the typical user. A large \(K_u\) thus represents a traffic hotspot, in which many users contend for the resources of the same serving satellites.
Let \(\delta\in[0,1]\) denote the concurrent-demand fraction, i.e., the fraction of associated candidate users with active downlink demand in the considered time-frequency resource block.
Since \(p_n\) denotes the shell-\(n\) association probability, the resource-sharing load of the shell-\(n\) serving satellite, including the typical user, is defined as
\begin{equation}\label{eq:shell_load}
L_n=1+\delta K_u p_n.
\end{equation}
The unit term accounts for the tagged user served by shell \(n\), whereas \(\delta K_up_n\) represents the shell-\(n\) share of the other candidate users with concurrent downlink demand. Candidate-user associations are approximated as independent categorical trials, where a candidate user associates with shell \(n\) with probability \(p_n\) and remains unserved with probability \(1-p_{\mathcal E}\).

For the considered time--frequency resource block, \(L_n\) is modeled as a deterministic, possibly noninteger resource-sharing factor rather than an instantaneous user count.
Define the signal-to-interference-plus-noise ratio (SINR) after \(i\) cancellation stages as
\begin{equation}\label{eq:sinr}
\mathrm{SINR}_n^{(i)}
=
\frac{\alpha P_n G_{\rm t} G_{\rm r} H_{0,n} R_n^{-2}}
{\mathcal I_n^{(i)}+\sigma^2}.
\end{equation}
Within each shell, users associated with the serving satellite are allocated equal time fractions over the full bandwidth \(W\), while different shells reuse the same bandwidth.
Thus, a user associated with shell \(n\) receives the time fraction \(1/L_n\), and its achievable rate after \(i\) cancellation stages is
\begin{equation}\label{eq:rate_n}
    \mathcal R_n^{(i)}
    =
    \frac{W}{L_n}
    \log_2\!\left(1+\mathrm{SINR}_n^{(i)}\right).
\end{equation}
The rate threshold \(\gamma\) is met at an admissible stage \(i\) when \(\mathcal R_n^{(i)}>\gamma\), or equivalently when \(\mathrm{SINR}_n^{(i)}>\tau_n\), where
\begin{equation}\label{eq:tau_n}
\tau_n=\tau_n(\gamma)
=
2^{L_n\gamma/W}-1.
\end{equation}
Since \(L_n\) depends on the association probability \(p_n\), the required SINR threshold is shell-dependent and indirectly controlled by the bias vector \(\mathbf B\).

The decoding threshold \(\tau_n\) is determined by the rate threshold and shell-\(n\) resource-sharing load and remains fixed throughout the SIC procedure.
The receiver first attempts direct decoding using \(\mathrm{SINR}_n^{(0)}\); if this fails, it retries the same codeword after successive interference cancellation.
Thus, SIC changes the decoding SINR through interference reduction, but not the rate threshold or required threshold.


\begin{figure}
    \centering
    \includegraphics[width=1\linewidth]{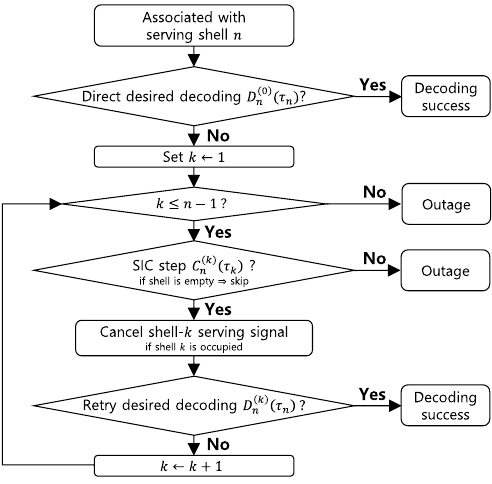}
    \caption{Block diagram of the SIC-aided decoding process for a user associated with shell $n$.}
    \label{fig:decoding_process}
\end{figure}

\subsection{Decoding Process}

Suppose that the typical user is associated with the serving satellite in shell \(n\). After the serving signals from shells \(1,\ldots,i\) have been canceled, the desired-link decodability event is
\begin{equation}\label{eq:D_event_new}
D_n^{(i)}(\tau_n)
=
\left\{
\mathrm{SINR}_n^{(i)}
>
\tau_n
\right\},
\end{equation}
where \(i=0,\ldots,n-1\). The case \(i=0\) corresponds to direct decoding without SIC. The noise power is \(\sigma^2=\kappa T_{\rm temp} W\), where \(\kappa\) is Boltzmann's constant, \(T_{\rm temp}\) is the system noise temperature, and $W$ is the bandwidth.

For \(k=1,\ldots,n-1\), let \(C_n^{(k)}(\tau_k)\) denote the event that shell \(k\) is empty or that its serving signal is decoded after the serving signals from shells \(1,\ldots,k-1\) have been canceled:
\begin{equation}\label{eq:C_event_new}
C_n^{(k)}(\tau_k)
=
\left\{
\Phi_k^{\rm serv}=\varnothing
\right\}
\cup
\left\{
\frac{
I_k^{\rm serv}
}{
\mathcal F_n^{(k)}+\sigma^2
}
>
\tau_k
\right\}.
\end{equation}
An empty lower shell is therefore skipped in the SIC procedure. The interference observed when decoding the shell-\(k\) serving signal is
\begin{equation}\label{eq:F_interference_new}
\mathcal F_n^{(k)}
=
\alpha P_n G_{\rm t}G_{\rm r}H_{0,n}R_n^{-2}
+
\mathcal I_n^{(k)}.
\end{equation}
Thus, while decoding the shell-\(k\) serving signal, the desired signal from shell \(n\), the uncanceled serving signals from shells \(k+1,\ldots,N\) excluding shell \(n\), and all non-serving satellite signals remain as interference. The threshold used in \(C_n^{(k)}\) is \(\tau_k=2^{L_k\gamma/W}-1\), since the signal being decoded is the serving signal of shell \(k\), whose satellite serves its own users at the same rate threshold \(\gamma\) under the shell-\(k\) load \(L_k\).
The overall decoding procedure is summarized in Fig.~\ref{fig:decoding_process}, where lower-shell serving signals are successively processed after direct decoding failure. For $n=1$, the procedure reduces to direct decoding.

For compact notation, we write \(D_n^{(i)}\) and \(C_n^{(k)}\) when the thresholds are clear. Unless otherwise stated, all probabilities, expectations, and Laplace transforms in this section are conditioned on \(E_n\), the event that the typical user is associated with shell \(n\); for example, \(\mathbb P(D_n^{(i)})\) denotes \(\mathbb P(D_n^{(i)}\mid E_n)\).


\begin{remark}[Shell-Wise SIC Ordering]
The SIC order is predetermined by shell altitude and proceeds from shell \(1\) to shell \(n-1\), rather than being determined by the instantaneous received powers.
This ordering exploits the systematically shorter propagation distances of lower-shell serving links, which tend to produce stronger inter-shell interference for users associated with upper shells, while avoiding realization-dependent interferer ranking.
Although fading and receive-antenna gains can alter the instantaneous power ordering, shell altitude still provides a geometry-driven large-scale cue for received-power ordering and thus a simple and stable basis for SIC.
SIC is purely a receiver-side enhancement and is not assumed in codeword or rate selection; under the ideal cancellation model, only successfully decoded signals are removed.
\end{remark}

\subsection{Mathematical Preliminaries}

This subsection introduces the conditional transforms and distributions used to evaluate the joint SIC events. Unless stated otherwise, all probabilities, expectations, and transforms are conditioned on the association event \(E_n\). For the residual interference we write \(\mathcal L_{\mathcal I_n^{(0)}\mid E_n,R_n}(s\mid r)=\mathbb E[e^{-s\mathcal I_n^{(0)}}\mid E_n,R_n=r]\) under the scalar conditioning used for direct decoding, and \(\mathcal L_{\mathcal I_n^{(i)}\mid E_n,\mathbf R}(s\mid\mathbf r)=\mathbb E[e^{-s\mathcal I_n^{(i)}}\mid E_n,\mathbf R=\mathbf r]\) for the joint SIC analysis, in which the complete serving-distance vector \(\mathbf R=\mathbf r\) is fixed.

The following link-level factors are obtained by averaging over the fading and antenna-gain variables of each interfering link.
The probability generating functional of the SPPP is then used to construct the aggregate interference transforms. For a non-serving satellite in shell \(m\) at distance \(v\),
\begin{align}\label{eq:M_non_joint}
M_m^{\rm non}(s,v)
&\triangleq
\mathbb E_H\!\left[
e^{-s \alpha P_mg_{\rm t}g_{\rm r}Hv^{-2}}
\right]
\nonumber\\
&=
\left(
1+\frac{s \alpha P_mg_{\rm t}g_{\rm r}v^{-2}}{m_0}
\right)^{-m_0}.
\end{align}
For a non-associated serving satellite,
\begin{align}
M_m^{\rm serv}(s,v)
&\triangleq
\mathbb E_{X_m,H}\!\left[
e^{-s \alpha P_mG_{\rm t}X_mHv^{-2}}
\right]
\nonumber\\
&=
p_{\rm r}
\left(
1+\frac{s\alpha P_mG_{\rm t}G_{\rm r}v^{-2}}{m_0}
\right)^{-m_0} 
\nonumber\\
&\quad+
(1-p_{\rm r})
\left(
1+\frac{s\alpha P_mG_{\rm t}g_{\rm r}v^{-2}}{m_0}
\right)^{-m_0}. \nonumber
\end{align}
The simple closed forms above follow from the Nakagami-\(m_0\) fading model, under which the power gain \(H\sim{\rm Gamma}(m_0,1/m_0)\) has Laplace transform \((1+s/m_0)^{-m_0}\).
Conditioned on the nearest visible satellite of shell \(m\) being at distance \(z<\infty\), the remaining shell-\(m\) satellites lie outside \(z\). Their aggregate non-serving interference has conditional transform
\begin{align}
\Psi_m^{\rm non}(s,z)
&\triangleq
\mathbb E\!\left[
e^{-sI_m^{\rm non}}\mid R_m=z
\right]
\nonumber\\
&=
\exp\!\left(
-\int_z^{r_{\max,m}}
\left[1-M_m^{\rm non}(s,v)\right]
d\Lambda_m(v)
\right). \nonumber
\end{align}

Conditioned on \(E_n\) and \(R_n=r\), the association rule excludes shell-\(m\) satellites closer than \(T_{mn}r\). Since feasible shell-\(m\) distances start at \(h_m\), define
\begin{equation}
\rho_{m,n}(r)\triangleq\max\{h_m,T_{mn}r\},
\qquad m\neq n. \nonumber
\end{equation}
By independent increments of the SPPP, the point process outside \(\rho_{m,n}(r)\) remains unchanged. For \(\rho_{m,n}(r)<z<r_{\max,m}\),
\begin{align}
\mathbb P(R_m>z\mid E_n,R_n=r)
&=
\exp\!\left(
-\Lambda_m(z)+\Lambda_m(\rho_{m,n}(r))
\right). \nonumber
\end{align}
The finite-distance component of the conditional serving-distance law is
\begin{align}\label{eq:conditional_Rm_finite_density_joint}
f_{R_m\mid E_n,R_n}^{\rm fin}(z\mid r)
&=
\mathds{1}_{\{\rho_{m,n}(r)<z<r_{\max,m}\}}
\Lambda_m'(z)
\nonumber\\
&\quad\times
\exp\!\left(
-\Lambda_m(z)+\Lambda_m(\rho_{m,n}(r))
\right).
\end{align}
The remaining admissible region may be empty; this case is represented by \(R_m=\infty\), with probability
\begin{align}\label{eq:void_probability_joint}
p_{m,n}^{\rm void}(r)
&\triangleq
\mathbb P(R_m=\infty\mid E_n,R_n=r)
\nonumber\\
&=
\exp\!\left(
-\Lambda_m(r_{\max,m})+\Lambda_m(\rho_{m,n}(r))
\right).
\end{align}
If \(\rho_{m,n}(r)\ge r_{\max,m}\), then \(p_{m,n}^{\rm void}(r)=1\) and the finite-distance density is zero. Since the finite component carries mass \(1-p_{m,n}^{\rm void}(r)\), \eqref{eq:conditional_Rm_finite_density_joint} and \eqref{eq:void_probability_joint} specify the conditional law of \(R_m\) without explicitly writing a Dirac mass at infinity, and conditioned on \(E_n\) and \(R_n=r\) the non-associated serving distances are independent across shells.

For direct decoding, all non-associated serving signals remain as interference. If \(R_m=\infty\), shell \(m\) contributes no interference. If \(R_m=z<\infty\), the nearest satellite contributes \(M_m^{\rm serv}(s,z)\), while the remaining satellites contribute \(\Psi_m^{\rm non}(s,z)\). Therefore,
\begin{align}\label{eq:direct_shell_transform_joint}
&\bar{\mathcal L}_{m\mid n}^{(0)}(s\mid r)
=
p_{m,n}^{\rm void}(r)
\nonumber\\
&+
\int_{\rho_{m,n}(r)}^{r_{\max,m}}
M_m^{\rm serv}(s,z)
\Psi_m^{\rm non}(s,z)
f_{R_m\mid E_n,R_n}^{\rm fin}(z\mid r)dz.
\end{align}
Consequently,
\begin{align}\label{eq:direct_residual_LT_joint}
\mathcal L_{\mathcal I_n^{(0)}\mid E_n,R_n}(s\mid r)
&=
\Psi_n^{\rm non}(s,r)
\prod_{m\neq n}
\bar{\mathcal L}_{m\mid n}^{(0)}(s\mid r).
\end{align}

\begin{lemma}[Direct-decoding contribution]\label{lem:direct_decoding_joint}
For integer \(m_0\ge 1\), the direct-decoding contribution is
\begin{align}\label{eq:direct_decoding_joint}
&q_{n,0}
=\mathbb P(D_n^{(0)}) = \int_{h_n}^{r_{\max,n}}
f_{R_n\mid E_n}(r)
\nonumber\\
&\times
\sum_{q=0}^{m_0-1}
\left.
\frac{(-s)^q}{q!}
\frac{d^q}{ds^q}
\left[
e^{-s\sigma^2}
\mathcal L_{\mathcal I_n^{(0)}\mid E_n,R_n}(s\mid r)
\right]
\right|_{s=s_n(r)}
dr,
\end{align}
where
\begin{equation}
s_n(r)
\triangleq
\frac{m_0\tau_n r^2}{\alpha P_nG_{\rm t}G_{\rm r}}. \nonumber
\end{equation}
\end{lemma}

\begin{proof}
Conditioned on \(E_n\), \(R_n=r\), and \(U=\mathcal I_n^{(0)}+\sigma^2\), direct decoding requires \(H_{0,n}>\tau_n r^2U/(\alpha P_nG_{\rm t}G_{\rm r})\). Using the integer-shape Gamma CCDF in \eqref{eq:ccdf gamma} and the identity \(\mathbb E[U^qe^{-sU}]=(-1)^q d^q\mathcal L_U(s)/ds^q\) gives the conditional summation in \eqref{eq:direct_decoding_joint}. Averaging over \eqref{eq:Rn_cond_En_pdf} completes the proof.
\end{proof}

For the joint SIC stages, let \(\mathbf R=(R_1,\ldots,R_N)\), where \(R_n=r_n<\infty\). For each \(m\neq n\), the conditional law of \(R_m\) is specified by \eqref{eq:void_probability_joint} and \eqref{eq:conditional_Rm_finite_density_joint}. If \(R_m=\infty\), shell \(m\) contributes neither serving nor non-serving interference. Define
\begin{equation}
Z_{n,i}
\triangleq
\sigma^2+\mathcal I_n^{(i)}. \nonumber
\end{equation}

\begin{lemma}[Full-distance-conditioned residual transform]\label{lem:joint_residual_transform}
Conditioned on \(E_n\) and \(\mathbf R=\mathbf r\), the Laplace transform of \(\mathcal I_n^{(i)}\) is
\begin{align}\label{eq:joint_residual_LT}
&\mathcal L_{\mathcal I_n^{(i)}\mid E_n,\mathbf R}(s\mid\mathbf r)
\nonumber\\
&=
\Psi_n^{\rm non}(s,r_n)
\prod_{\substack{m\neq n\\r_m<\infty}}
\Psi_m^{\rm non}(s,r_m)
\prod_{\substack{m=i+1\\m\neq n,\ r_m<\infty}}^{N}
M_m^{\rm serv}(s,r_m).
\end{align}
Consequently,
\begin{align}\label{eq:joint_noise_inclusive_LT}
\mathcal L_{Z_{n,i}\mid E_n,\mathbf R}(s\mid\mathbf r)
&=
e^{-s\sigma^2}
\mathcal L_{\mathcal I_n^{(i)}\mid E_n,\mathbf R}(s\mid\mathbf r).
\end{align}
\end{lemma}

\begin{proof}
    See Appendix \ref{proof:joint_residual_transform}.
\end{proof}


The transform in Lemma~\ref{lem:joint_residual_transform} characterizes the residual interference in the transform domain. To evaluate the joint decoding event, we need the probability that the noise-inclusive residual term lies in an interval. 
For a fixed $\mathbf r$, define
\begin{align}
F_{\mathcal I_n^{(i)}|\mathbf r}(\xi)
&\triangleq
\mathbb{P}\!\left(
\mathcal I_n^{(i)}\le \xi
\,\middle|\,
E_n,\mathbf R=\mathbf r
\right),\; \mathrm{and} \\
F_{Z_{n,i}|\mathbf r}(z)
&\triangleq
\mathbb{P}\!\left(
Z_{n,i}\le z
\,\middle|\,
E_n,\mathbf R=\mathbf r
\right).
\end{align}

\begin{lemma}[Conditional distribution of the residual interference]\label{lem:conditional_residual_cdf}
Under the continuous Nakagami fading model, the conditional cumulative distribution function (CDF) of \(\mathcal I_n^{(i)}\) is
\begin{align}\label{eq:gil_pelaez_joint_cdf}
&F_{\mathcal I_n^{(i)}\mid\mathbf r}(\xi)
\nonumber\\
&=
\begin{cases}
0, & \xi<0,\\[1mm]
\displaystyle
\lim_{s\to+\infty}
\mathcal L_{\mathcal I_n^{(i)}\mid E_n,\mathbf R}
(s\mid\mathbf r),
& \xi=0,\\[3mm]
\displaystyle
\frac{1}{2}
-
\frac{1}{\pi}
\int_0^\infty
\frac{
\operatorname{Im}\!\left[
e^{-\mathrm j\omega\xi}
\mathcal L_{\mathcal I_n^{(i)}\mid E_n,\mathbf R}
(-\mathrm j\omega\mid\mathbf r)
\right]
}{\omega}
d\omega,
& \xi>0.
\end{cases}
\end{align}
The corresponding noise-inclusive CDF is
\begin{align}\label{eq:joint_noise_shifted_cdf}
F_{Z_{n,i}\mid\mathbf r}(z)
=
\begin{cases}
0, & z<\sigma^2,\\[1mm]
F_{\mathcal I_n^{(i)}\mid\mathbf r}(z-\sigma^2), & z\ge \sigma^2.
\end{cases}
\end{align}
\end{lemma}

\begin{proof}
See Appendix \ref{proof:conditional_residual_cdf}.
\end{proof}



The CDF in Lemma~\ref{lem:conditional_residual_cdf} is an intermediate quantity rather than a decoding probability. In the main result, conditioning on the desired and serving-interference powers converts the desired-decoding and SIC inequalities into a single interval constraint on \(Z_{n,i}\), whose probability is evaluated using \(F_{Z_{n,i}\mid\mathbf r}\).

\subsection{Main Results on Rate Coverage}

Define the cumulative SIC-success event as
\begin{equation}
\mathcal C_n^{(i)}
\triangleq
\bigcap_{k=1}^{i} C_n^{(k)}, \quad i=1,\ldots,n-1, \nonumber
\end{equation}
with \(\mathcal C_n^{(0)}\triangleq\Omega\), where $\Omega$ is the sure event. Define the incremental desired-decoding event as
\begin{equation}
\mathcal D_n^{(i)}
\triangleq
D_n^{(i)}\cap\left(D_n^{(i-1)}\right)^c, \quad i=1,\ldots,n-1, \nonumber
\end{equation}
with \(\mathcal D_n^{(0)}\triangleq D_n^{(0)}\). Since each SIC stage removes a nonnegative serving-interference term,
\begin{equation}\label{eq:D_nested_property_new}
D_n^{(0)}\subseteq D_n^{(1)}\subseteq\cdots\subseteq D_n^{(n-1)}.
\end{equation}
Thus, \(\{\mathcal D_n^{(i)}\}_{i=0}^{n-1}\) are mutually disjoint, and the shell-\(n\) successful decoding event is
\begin{equation}\label{eq:shell_n_success_event_joint}
\mathcal S_n(\gamma)
=
\bigcup_{i=0}^{n-1}
\left(\mathcal D_n^{(i)}\cap\mathcal C_n^{(i)}\right).
\end{equation}

The term \(i=0\) corresponds to direct decoding, while \(i\ge1\) corresponds to first decoding the desired signal after the \(i\)-th SIC stage with all preceding SIC steps successful. We write \(\mathcal P_{{\rm cov},n}(\gamma)\triangleq\mathbb P(\mathcal S_n(\gamma)\mid E_n)\) for the shell-conditioned rate coverage, which is averaged over the shell-wise association probabilities in Corollary~\ref{coro:rate coverage}.

To evaluate these stage-wise events, define the desired received power and lower-shell serving powers, conditioned on \(\mathbf R=\mathbf r\), as
\begin{align}
\mathsf S_n
&\triangleq
\alpha P_nG_{\rm t}G_{\rm r}H_{0,n}r_n^{-2}, \nonumber 
\\
Y_k
&\triangleq
\begin{cases}
\alpha P_kG_{\rm t}X_kH_{0,k}r_k^{-2}, & r_k<\infty,\\
0, & r_k=\infty,
\end{cases}
\quad k<n. \nonumber
\end{align}
For compactness, let
\begin{equation}
g_{m_0}(x;\mu)
\triangleq
\frac{m_0^{m_0}}{\Gamma(m_0)\mu^{m_0}}
x^{m_0-1}e^{-m_0x/\mu}, \quad x>0, \nonumber
\end{equation}
denote the Gamma density with shape \(m_0\) and mean \(\mu\). Then
\begin{equation}\label{eq:desired_power_density_joint}
f_{\mathsf S_n\mid r_n}(x)
=
g_{m_0}\!\left(x;\alpha P_nG_{\rm t}G_{\rm r}r_n^{-2}\right),
\end{equation}
and, for \(r_k<\infty\),
\begin{align}\label{eq:serving_power_density_joint}
&f_{Y_k\mid R_k}(y\mid r_k)
=
p_{\rm r}
g_{m_0}\!\left(y;\alpha P_kG_{\rm t}G_{\rm r}r_k^{-2}\right)
\nonumber\\
&\phantom{------}+
(1-p_{\rm r})
g_{m_0}\!\left(y;\alpha P_kG_{\rm t}g_{\rm r}r_k^{-2}\right),
\quad y>0.
\end{align}
When \(r_k=\infty\), \(Y_k=0\) almost surely. Conditioned on \(E_n\) and \(\mathbf R=\mathbf r\), \(\mathsf S_n\), \(\{Y_k\}_{k=1}^{i}\), and \(Z_{n,i}\) are mutually independent because they depend on disjoint fading variables, antenna-gain variables, and point-process components.

\begin{lemma}[Residual-interval representation]\label{lem:joint_interval_representation}
Define
\begin{equation}
F_{Z_{n,i}\mid\mathbf r}^{<}(t)
\triangleq
\mathbb P(Z_{n,i}<t\mid E_n,\mathbf R=\mathbf r). \nonumber
\end{equation}
Since \(\mathcal I_n^{(i)}\) is continuous on \((0,\infty)\), \(F^{<}_{Z_{n,i}\mid\mathbf r}(t)=0\) for \(t\le\sigma^2\) and \(F^{<}_{Z_{n,i}\mid\mathbf r}(t)=F_{Z_{n,i}\mid\mathbf r}(t)\) for \(t>\sigma^2\).
For fixed \(\mathsf S_n=x\), \(\mathbf Y_i=\mathbf y_i=(y_1,\ldots,y_i)\), and \(\mathbf R=\mathbf r\),
\begin{equation}\label{eq:joint_event_interval_equivalence}
\mathcal D_n^{(i)}\cap\mathcal C_n^{(i)}
=
\left\{
\ell_{n,i}(x,\mathbf y_i)
\le
Z_{n,i}
<
u_{n,i}(x,\mathbf y_i,\mathbf r)
\right\},
\end{equation}
where
\begin{equation}\label{eq:joint_interval_lower}
\ell_{n,i}(x,\mathbf y_i)
\triangleq
\frac{x}{\tau_n}-y_i
\end{equation}
and
\begin{align}\label{eq:joint_interval_upper}
u_{n,i}(x,\mathbf y_i,\mathbf r)
&\triangleq
\min\Bigg\{
\frac{x}{\tau_n},
\min_{\substack{1\le k\le i\\r_k<\infty}}
\left[
\frac{y_k}{\tau_k}
-x
-\sum_{j=k+1}^{i}y_j
\right]
\Bigg\}.
\end{align}
The minimum over an empty index set is interpreted as \(+\infty\). Consequently, we have
\begin{align}\label{eq:joint_residual_interval_probability}
\mathcal G_{n,i}(x,\mathbf y_i\mid\mathbf r)
&=
\mathds{1}_{\{u_{n,i}>\ell_{n,i}\}}
\left[
F_{Z_{n,i}\mid\mathbf r}^{<}(u_{n,i})
-
F_{Z_{n,i}\mid\mathbf r}^{<}(\ell_{n,i})
\right].
\end{align}
\end{lemma}
\begin{proof} 
See Appendix \ref{proof:joint_interval_representation}.
\end{proof}

Lemma~\ref{lem:joint_interval_representation} reduces the joint desired-decoding and SIC conditions to an interval constraint on the common residual term \(Z_{n,i}\). The interval probability is then obtained from the residual CDF derived in Lemma~\ref{lem:conditional_residual_cdf}.

\begin{theorem}[Shell-conditioned rate coverage]\label{thm:joint_analytical_coverage}
Conditioned on association with shell \(n\), the rate coverage probability is
\begin{equation}\label{eq:exact_joint_shell_coverage}
\mathcal{P}_{\mathrm{cov},n}(\gamma)
=
q_{n,0}
+
\sum_{i=1}^{n-1}q_{n,i},
\end{equation}
where \(q_{n,0}\) is given in Lemma~\ref{lem:direct_decoding_joint}. For \(i=1,\ldots,n-1\),
\begin{equation}
q_{n,i}
\triangleq
\mathbb P\!\left(
\mathcal D_n^{(i)}
\cap
\mathcal C_n^{(i)}
\mid E_n
\right), \nonumber
\end{equation}
and
\begin{align}\label{eq:exact_joint_stage_expectation}
&q_{n,i}
=
\mathbb E\!\left[
\mathcal G_{n,i}
\left(
\mathsf S_n,\mathbf Y_i
\mid\mathbf R
\right)
\mid E_n
\right]
\nonumber\\
&=
\int_{h_n}^{r_{\max,n}}
f_{R_n\mid E_n}(r_n)
\mathbb E_{\mathbf R_{-n}\mid E_n,R_n=r_n}
\Bigg[
\int_0^\infty
f_{\mathsf S_n\mid r_n}(x)
\nonumber\\
&\phantom{----------}\times
\mathbb E_{\mathbf Y_i\mid\mathbf R}
\left[
\mathcal G_{n,i}
\left(
x,\mathbf Y_i
\mid\mathbf R
\right)
\right]
dx
\Bigg]
dr_n.
\end{align}

Here, \(\mathbf R_{-n}\) denotes the serving-distance vector excluding \(R_n\), and \(\mathbf Y_i=(Y_1,\ldots,Y_i)\). For each \(m\neq n\), the conditional law of \(R_m\) is specified by \eqref{eq:void_probability_joint} and \eqref{eq:conditional_Rm_finite_density_joint}. For \(r_k<\infty\), \(Y_k\) has density \eqref{eq:serving_power_density_joint}; for \(r_k=\infty\), \(Y_k=0\) almost surely.
\end{theorem}

\begin{proof}
For \(i=1,\ldots,n-1\), applying the law of total expectation to the stage-\(i\) event gives
\begin{align}
q_{n,i}
&=
\mathbb E\Big[
\mathbb P\left(
\mathcal D_n^{(i)}\cap\mathcal C_n^{(i)}
\mid
E_n,\mathbf R,\mathsf S_n,\mathbf Y_i
\right)
\mid E_n
\Big]. \nonumber
\end{align}
By Lemma~\ref{lem:joint_interval_representation}, the inner conditional probability equals \(\mathcal G_{n,i}(\mathsf S_n,\mathbf Y_i\mid\mathbf R)\), which proves the first equality in \eqref{eq:exact_joint_stage_expectation}. Expanding this expectation over \(R_n\), \(\mathbf R_{-n}\), \(\mathsf S_n\), and \(\mathbf Y_i\), using the conditional independence stated above and the distributions in \eqref{eq:Rn_cond_En_pdf}, \eqref{eq:void_probability_joint}, \eqref{eq:conditional_Rm_finite_density_joint}, \eqref{eq:desired_power_density_joint}, and \eqref{eq:serving_power_density_joint}, gives the second equality in \eqref{eq:exact_joint_stage_expectation}.

It remains to sum the mutually disjoint stage-wise events. From \eqref{eq:D_nested_property_new}, the incremental desired-decoding events \(\{\mathcal D_n^{(i)}\}_{i=0}^{n-1}\) are mutually disjoint, and hence so are \(\{\mathcal D_n^{(i)}\cap\mathcal C_n^{(i)}\}_{i=0}^{n-1}\). Using \eqref{eq:shell_n_success_event_joint},
\begin{align}
\mathcal{P}_{\mathrm{cov},n}(\gamma)
&=
\mathbb P(\mathcal S_n(\gamma)\mid E_n)
\nonumber\\
&=
\sum_{i=0}^{n-1}
\mathbb P\!\left(
\mathcal D_n^{(i)}\cap\mathcal C_n^{(i)}
\mid E_n
\right). \nonumber
\end{align}
The \(i=0\) term is \(q_{n,0}\), and the remaining terms are \(q_{n,i}\), which proves \eqref{eq:exact_joint_shell_coverage}. For \(n=1\), the summation is empty and \(\mathcal{P}_{\mathrm{cov},1}(\gamma)=q_{1,0}\).
\end{proof}

\begin{corollary}\label{coro:rate coverage}
The overall rate coverage probability is
\begin{equation}\label{eq:exact_joint_overall_coverage}
p_{\rm cov}(\gamma)
=
\sum_{n=1}^{N}
 \mathcal{P}_{\mathrm{cov},n}(\gamma) p_n.
\end{equation}
\end{corollary}

\begin{proof}
The association events \(\{E_n\}_{n=1}^{N}\) are mutually disjoint, and the typical user is in outage when no visible serving satellite is available. Averaging the shell-conditioned rate coverage in Theorem~\ref{thm:joint_analytical_coverage} over the shell-wise association probabilities in Theorem~\ref{thm:shell_wise_association_probability} gives the result.
\end{proof}

Corollary~\ref{coro:rate coverage} completes the analytical characterization of rate coverage. The expression combines the shell-wise association probabilities, the load-dependent SINR thresholds, and the SIC-aided decoding probabilities, and is evaluated numerically in the next section.

\begin{remark}[Biasing and SIC are Mutually Reinforcing]
Association biasing affects not only the serving-shell selection but also the decodability of lower-shell serving signals during SIC.
Decoding shell \(k\) requires its SINR to exceed \(\tau_k=2^{L_k\gamma/W}-1\), where \(L_k=1+\delta K_up_k\).
Hence, biasing traffic away from a lower shell reduces its association probability \(p_k\), lowers its decoding threshold, and facilitates cancellation of its serving signal.
Conversely, SIC mitigates the inter-shell interference experienced by users offloaded to upper shells, partially compensating for their larger propagation loss.
Thus, association biasing and SIC reinforce each other through the coupling among shell association, load, and inter-shell interference. 
This coupling among association, load, and SIC is consistent with the SIC-aware association analysis in \cite{wildemeersch:tcom:2014}.
\end{remark}

\begin{table}[t!]
\caption{Simulation parameters.}
\label{tab:simulation_parameters}
\centering
\begin{tabular}{ll}
\toprule
Parameter & Value \\
\midrule
Earth radius & $R_\oplus = 6371\text{ km}$ \\
Carrier frequency & $f_{\mathrm{c}} = 12\text{ GHz}$ \\
Transmit power & $P = 35\text{ dBm}$ \\
System bandwidth & $W = 250\text{ MHz}$ \\
Transmit antenna gains & $G_{\mathrm{t}} = 35\text{ dBi}$, $g_{\mathrm{t}} = 15\text{ dBi}$ \\
Receive antenna gains & $G_{\mathrm{r}} = 10\text{ dBi}$, $g_{\mathrm{r}} = 0\text{ dBi}$ \\
Main-lobe reception probability & $p_{\mathrm{r}} = 0.2$ \\
Concurrent-demand fraction & $\delta = 0.1$ \\
Minimum elevation angle & $\theta_{\min} = 15^\circ$ \\
Noise temperature & $T_{\rm{temp}} = 120\text{ K}$ \\
\bottomrule
\end{tabular}
\end{table}

\begin{table*}[t!]
\caption{Comparison of the proposed and baseline configurations.}
\label{tab:baseline_comparison}
\centering
\footnotesize
\renewcommand{\arraystretch}{1.0}
\begin{tabularx}{\textwidth}{@{} l >{\raggedright\arraybackslash}X l c l c @{}}
\toprule
Configuration
& Deployment and serving links
& Traffic allocation
& Bandwidth per link
& Treatment of other serving signals
& SIC \\
\midrule
\makecell[tl]{Proposed}
& \makecell[tl]{Multi-shell; nearest serving satellite per\\shell}
& \makecell[tl]{Biased received power}
& \makecell[tc]{\(W\)}
& \makecell[tl]{Co-channel under full reuse}
& \makecell[tc]{Yes} \\
\makecell[tl]{MS-TIN\\\strut}
& \makecell[tl]{As proposed}
& \makecell[tl]{As proposed}
& \makecell[tc]{\(W\)}
& \makecell[tl]{Treated as noise}
& \makecell[tc]{No} \\
\makecell[tl]{MS-OMA (\(N\))}
& \makecell[tl]{As proposed}
& \makecell[tl]{Uniform over \(N\) links;\\over \(M\) if \(M<N\)}
& \makecell[tc]{\(W/N\);\\\(W/M\) if \(M<N\)}
& \makecell[tl]{Orthogonalized across links}
& \makecell[tc]{No} \\
\makecell[tl]{SS-OMA (\(N\)), \(h\)}
& \makecell[tl]{Single shell at altitude \(h\); up to \(N\)\\nearest visible satellites}
& \makecell[tl]{Uniform over \(N\) links;\\over \(M\) if \(M<N\)}
& \makecell[tc]{\(W/N\);\\\(W/M\) if \(M<N\)}
& \makecell[tl]{Orthogonalized across\\active links}
& \makecell[tc]{No} \\
\bottomrule
\end{tabularx}
\end{table*}

\section{Numerical Results}\label{sec:numerical results}

\begin{figure}[t!]
    \centering
    \includegraphics[width=1\linewidth]{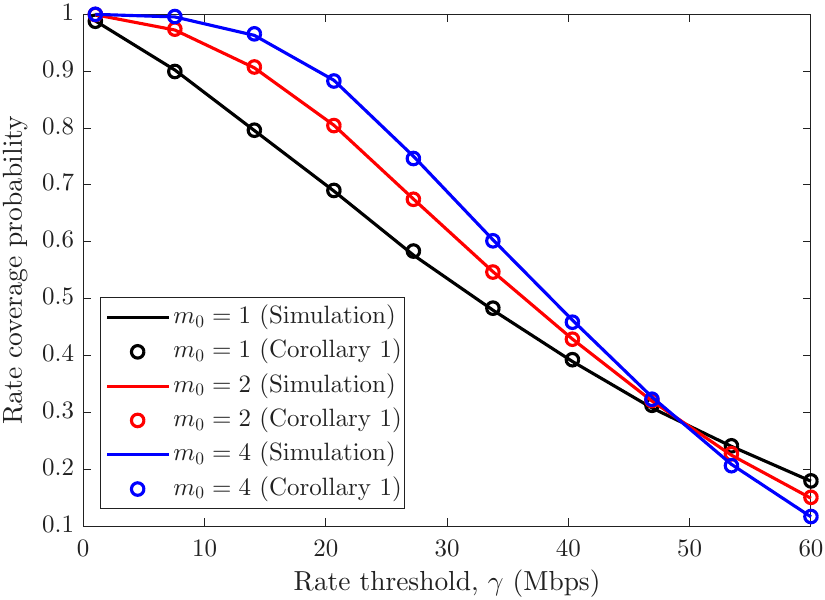}
    \caption{Analytical and simulated rate coverage probability for different values of the Nakagami parameter \(m_0\), with \(\boldsymbol{\Theta}=[1000,1000,1000]\), \(\mathbf B_{\mathrm{dB}}=[0,1.7,2.4]~\mathrm{dB}\), and \(K_u=250\).}
    \label{fig:validation}
\end{figure}

Unless otherwise specified, a three-shell configuration with \(N=3\) and \(\mathbf h=[360,480,550]~\mathrm{km}\) is considered.
In the numerical section, the bias vector in decibels is denoted by \(\mathbf B_{\mathrm{dB}}=[10\log_{10}B_1,\dots,10\log_{10}B_N]\), and \(\boldsymbol{\Theta}=[4\pi\lambda_1a_1^2,\ldots,4\pi\lambda_Na_N^2]\) denotes the average satellite count vector.
To isolate the effects of shell geometry and association biasing, we use a common transmit power \(P_n=P=35~\mathrm{dBm}\) across shells as a reference configuration, while practical systems may employ shell-dependent transmit powers.
Experiment-specific satellite-count and bias settings are given with the corresponding results.
The remaining parameters are listed in Table~\ref{tab:simulation_parameters}.

The three baselines in Table~\ref{tab:baseline_comparison} are multi-shell treating interference as noise (MS-TIN), multi-shell orthogonal multiple access (MS-OMA), and single-shell orthogonal multiple access (SS-OMA); MS-TIN is the proposed architecture without SIC.
The OMA baselines use \(N\) orthogonal serving links when available; if only \(M<N\) links are available, all \(M\) are used, while \(M=0\) results in outage.
MS-OMA selects one nearest visible satellite per occupied shell, whereas SS-OMA selects up to the \(N\) nearest visible satellites from one shell.
The active links equally share bandwidth and traffic, yielding \(W/M\), \(L_{\rm OMA}(M)=1+\delta K_u/M\), and thermal noise over \(W/M\).
Non-serving satellites are partitioned among the \(M\) active groups with group sizes differing by at most one, and rate coverage is averaged over the \(M\) active links.
For SS-OMA with \(h=h_n\), the average satellite count matches that of shell \(n\).

Fig.~\ref{fig:validation} compares the analytical rate coverage probability in Corollary~\ref{coro:rate coverage} with Monte Carlo simulations for different Nakagami parameters \(m_0\).
The close agreement over the considered rate threshold range validates the analytical characterization of shell association, residual interference, and SIC-aided rate coverage.

\begin{figure}[t!]
    \centering
    \includegraphics[width=1\linewidth]{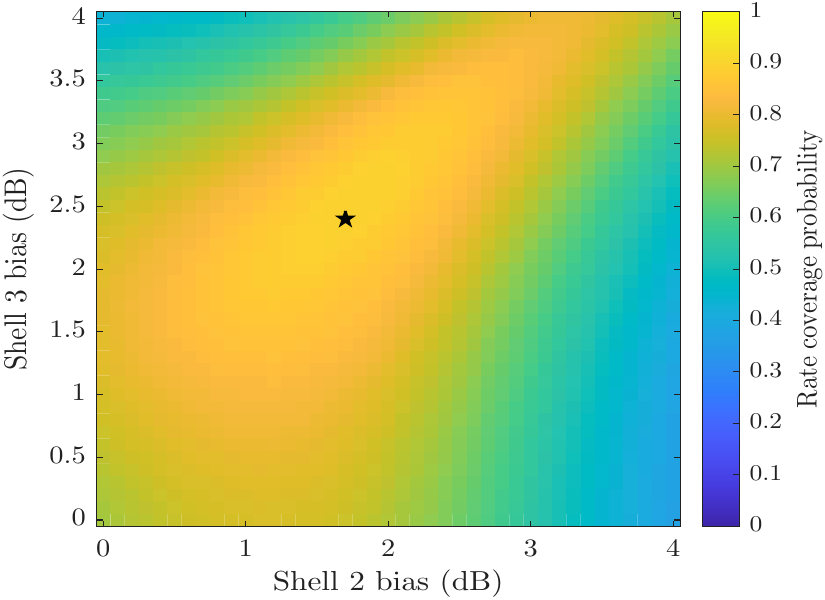}
    \caption{Rate coverage probability versus the shell-2 and shell-3 association biases for \(\boldsymbol{\Theta}=[1000,1000,1000]\). The star marks the grid-optimal bias vector \(\mathbf B_{\mathrm{dB}}^\star=[0,1.7,2.4]~\mathrm{dB}\).}
    \label{fig:bias}
\end{figure}
\begin{figure}[t!]
    \centering
    \includegraphics[width=1\linewidth]{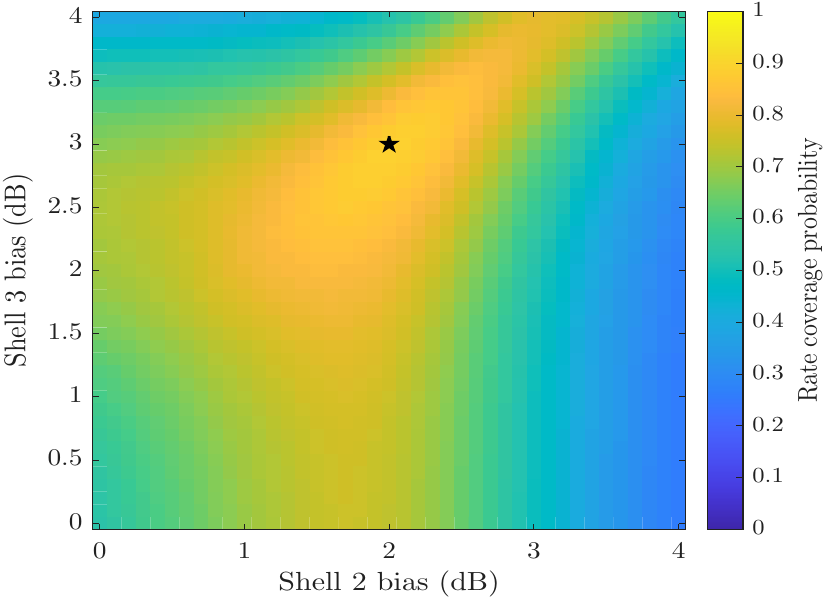}
    \caption{Rate coverage probability versus the shell-2 and shell-3 association biases for \(\boldsymbol{\Theta}=[2000,2000,2000]\). The star marks the grid-optimal bias vector \(\mathbf B_{\mathrm{dB}}^\star=[0,2,3]~\mathrm{dB}\). Compared with Fig.~\ref{fig:bias}, the larger average satellite count shifts the optimum toward stronger upper-shell biasing.}
    \label{fig:bias2}
\end{figure}

\subsection{The Role of Bias}

\begin{figure}
    \centering
    \includegraphics[width=1\linewidth]{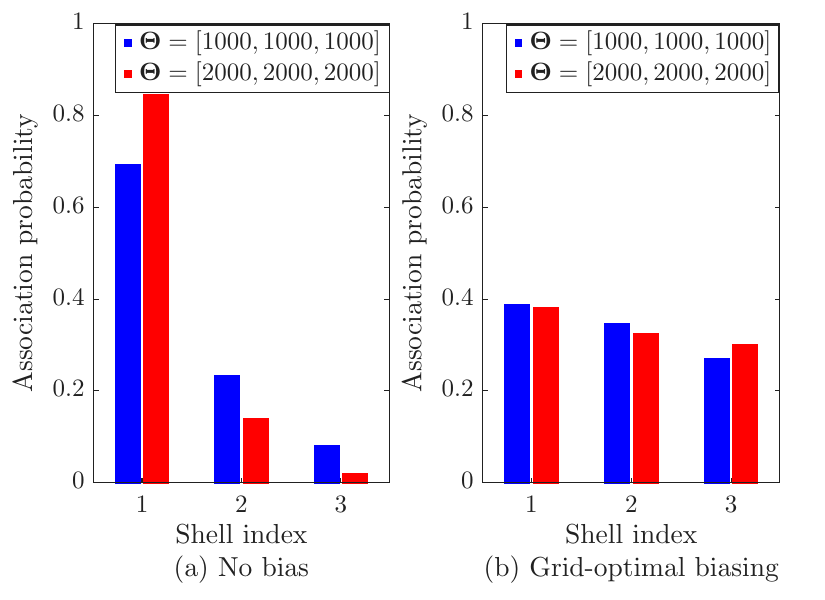}
    \caption{Shell association probabilities: (a) without biasing and (b) with grid-optimal biasing. The corresponding \(\mathbf B_{\mathrm{dB}}^\star\) vectors are \([0,1.7,2.4]\) and \([0,2,3]~\mathrm{dB}\) for blue and red, respectively.}
    \label{fig:load}
\end{figure}

\begin{figure}
    \centering
    \includegraphics[width=1\linewidth]{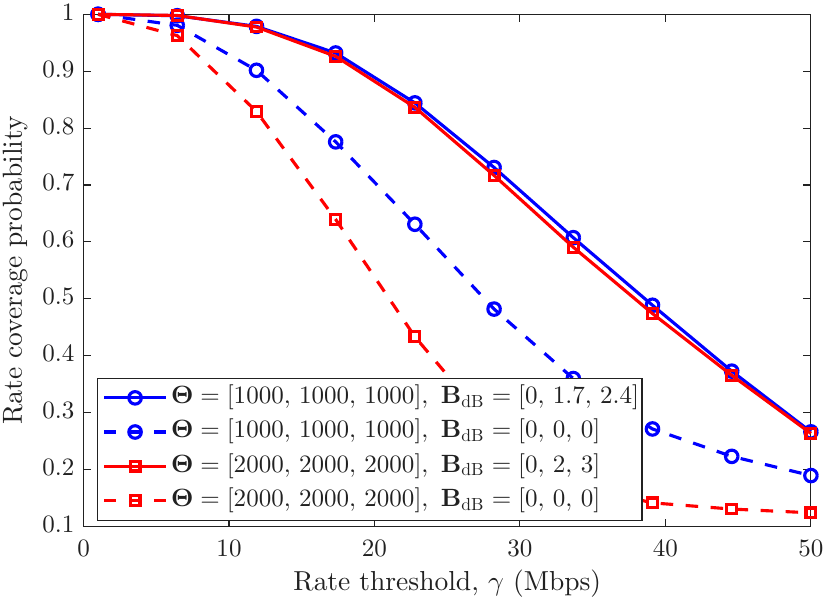}
    \caption{Rate coverage probability versus rate threshold for \(\boldsymbol{\Theta}=[1000,1000,1000]\) and \(\boldsymbol{\Theta}=[2000,2000,2000]\).}
    \label{fig:number compare}
\end{figure}

Figs.~\ref{fig:bias} and~\ref{fig:bias2} show the rate coverage probability obtained by sweeping the biases of shells 2 and 3, with the first entry of \(\mathbf B_{\mathrm{dB}}\) fixed at \(0~\mathrm{dB}\), \(m_0=4\), \(K_u=250\), and \(\gamma=20~\mathrm{Mbps}\).
For each satellite count vector, the bias vector that maximizes the rate coverage probability over the selected grid is denoted by \(\mathbf B_{\mathrm{dB}}^\star\) and referred to as the optimal bias for the considered operating point.
The grid-optimal bias therefore shifts toward stronger upper-shell biasing as the constellation becomes denser.
These biases are reused without re-optimization: Fig.~\ref{fig:number compare} uses the corresponding vector, while Figs.~\ref{fig:SIC and single},~\ref{fig:increase user 101010} and ~\ref{fig:pr compare 101010} use the \(\boldsymbol{\Theta}=[1000,1000,1000]\) vector.

Fig.~\ref{fig:load} explains this density-dependent bias through the shell association probabilities.
Without biasing, the denser constellation associates more users with shell 1 because the shorter serving distances increase its received-power advantage over the upper shells.
Applying the optimal bias shifts traffic toward shells 2 and 3 in both configurations.
Nevertheless, shell 1 retains a slightly larger association probability than the upper shells, allowing more users to exploit the shorter propagation distance of lower-altitude service.
The association probabilities predicted by Theorem~\ref{thm:shell_wise_association_probability} closely agree with the Monte Carlo results. Across the three shell indices, the two satellite-count configurations, and both biasing conditions, the maximum squared error is \(3.95\times10^{-5}\).

The association imbalance in Fig.~\ref{fig:load} directly translates into unequal shell loads, as reflected in the rate coverage probability in Fig.~\ref{fig:number compare} for \(m_0=4\) and \(K_u=250\).
Without biasing, the denser constellation places more users on shell 1, increasing its load and the corresponding SINR threshold.
These effects result in lower rate coverage despite the shorter serving distances.
After biasing, load redistribution largely removes the difference caused by shell-1 concentration.
Since each shell provides at most one serving link, densification does not increase the maximum number of serving links within the footprint.
A network-wide benefit can instead arise from supporting more footprints, which is not captured by the typical-footprint rate coverage considered here.

\textbf{Takeaway 1 (Biasing and Lower-Shell Preference):} Multi-shell geometry tends to favor lower-shell association, while shell-dependent biasing redistributes traffic toward the upper shells. This load balancing improves rate coverage while preserving some lower-shell preference to exploit shorter propagation distances.

\textbf{Takeaway 2 (Per-Footprint Limit of Densification):} Increasing the satellite count does not necessarily improve per-footprint rate coverage because it does not add shell-wise serving opportunities. A network-wide benefit can instead arise from supporting more footprints.

\begin{figure}
    \centering
    \includegraphics[width=1\linewidth]{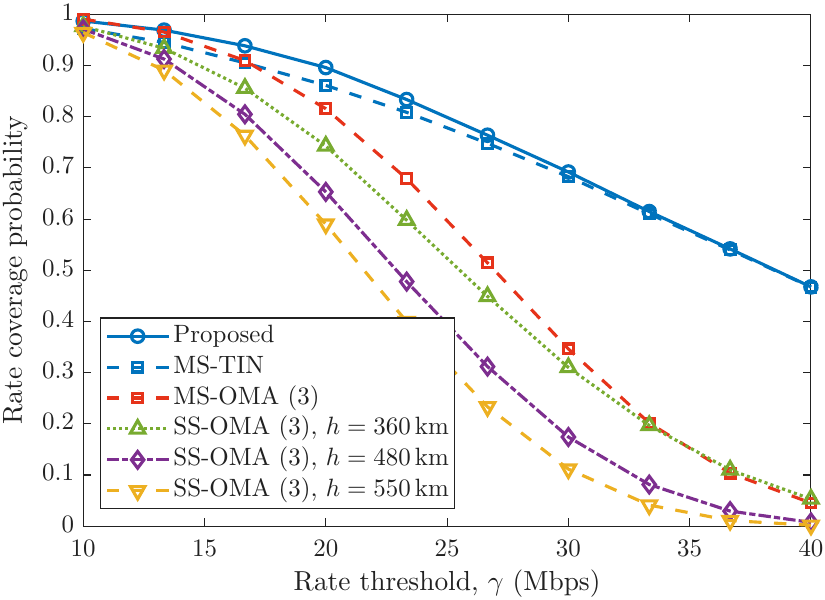}
    \caption{Rate coverage probability versus rate threshold for the proposed architecture and baseline configurations, with \(\boldsymbol{\Theta}=[1000,1000,1000]\) and \(\mathbf B_{\mathrm{dB}}=[0,1.7,2.4]~\mathrm{dB}\).}
    \label{fig:SIC and single}
\end{figure}

\begin{figure}
    \centering
    \includegraphics[width=1\linewidth]{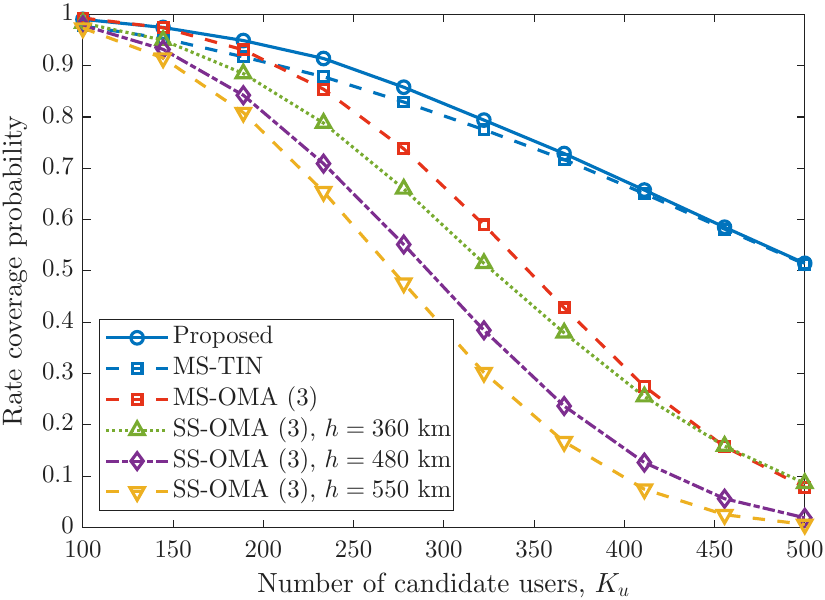}
    \caption{Rate coverage probability versus the number of candidate users \(K_u\) for the proposed architecture and baseline configurations, with \(\boldsymbol{\Theta}=[1000,1000,1000]\), \(\mathbf B_{\mathrm{dB}}=[0,1.7,2.4]~\mathrm{dB}\), and \(\gamma=20~\mathrm{Mbps}\).}
    \label{fig:increase user 101010}
\end{figure}

\begin{figure}
    \centering
    \includegraphics[width=1\linewidth]{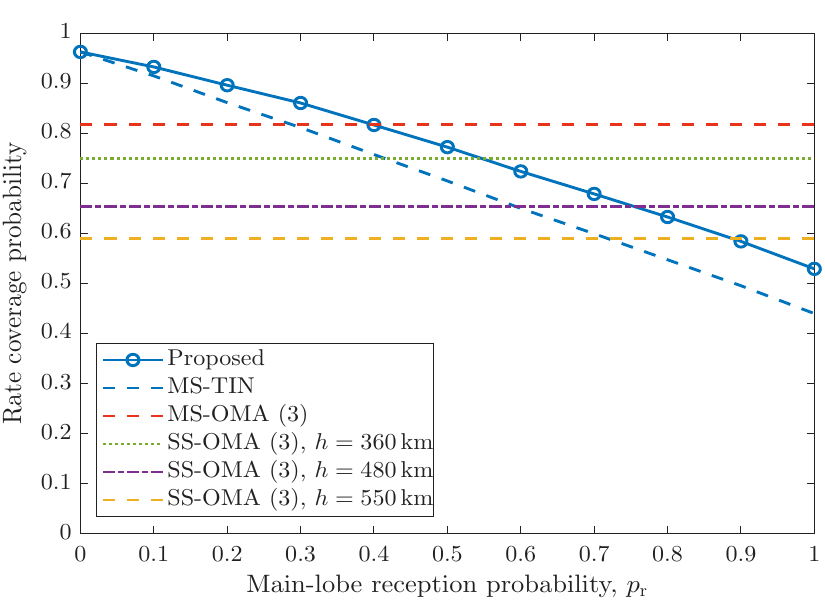}
    \caption{Rate coverage probability versus the main-lobe reception probability \(p_{\rm r}\) with \(\boldsymbol{\Theta}=[1000,1000,1000]\) and \(\mathbf B_{\mathrm{dB}}=[0,1.7,2.4]~\mathrm{dB}\).}
    \label{fig:pr compare 101010}
\end{figure}

\subsection{Comparison With Baselines}

Fig.~\ref{fig:SIC and single} compares the proposed architecture with the baselines for \(m_0=4\) and \(K_u=250\).
The proposed architecture achieves the highest rate coverage probability, while MS-TIN differs only by the absence of SIC, so their gap quantifies the gain from SIC.
As \(\gamma\) increases, the higher decoding threshold reduces the decodability of lower-shell serving signals, narrowing the gap to MS-TIN.
Despite its threefold larger total average satellite count, MS-OMA provides only a limited advantage over SS-OMA at \(h=360~\mathrm{km}\), which diminishes at higher rate thresholds. 
Among the SS-OMA baselines, the lowest shell performs best due to its shorter serving distances.

Fig.~\ref{fig:increase user 101010} shows the rate coverage probability as \(K_u\) increases.
Increasing \(K_u\) increases the resource-sharing load and hence the required SINR threshold, producing a trend similar to that observed with increasing \(\gamma\) in Fig.~\ref{fig:SIC and single}.
At low \(K_u\), all configurations achieve relatively high rate coverage probability because the resource-sharing load is small.
As \(K_u\) increases, the rate coverage probability of the OMA baselines decreases more rapidly because of bandwidth partitioning and increasing load.
The proposed architecture maintains higher rate coverage probability by distributing traffic across shells while retaining bandwidth \(W\) on each serving link.

\textbf{Takeaway 3 (Benefit in Traffic Hotspots):} Multi-shell load balancing with full frequency reuse improves rate coverage probability across traffic loads, and its advantage becomes more pronounced as traffic demand increases.

Fig.~\ref{fig:pr compare 101010} varies \(p_{\rm r}\) with \(m_0=4\), \(K_u=250\), and \(\gamma=20~\mathrm{Mbps}\).
As \(p_{\rm r}\) increases, more non-associated serving signals are received with main-lobe gain, increasing inter-shell interference and reducing the rate coverage probability of the full-reuse configurations, while the OMA baselines remain unchanged.
The proposed architecture outperforms all baselines up to approximately \(p_{\rm r}=0.4\), which gives the tolerable main-lobe reception probability for full frequency reuse at the considered operating point.
Beyond this crossover, stronger inter-shell interference offsets the benefit of full frequency reuse, while the widening gap to MS-TIN shows that SIC becomes more valuable as \(p_{\rm r}\) increases.
Figs.~\ref{fig:SIC and single} and~\ref{fig:increase user 101010} further show that the rate coverage advantage of full frequency reuse over OMA increases with the rate threshold and traffic load.
This suggests that the crossover can shift to larger \(p_{\rm r}\) under heavier traffic or higher rate requirements, allowing a less stringent receive-side isolation requirement.

\textbf{Takeaway 4 (Receive Isolation and SIC):} Increasing \(p_{\rm r}\) weakens receive-side isolation and reduces the rate coverage advantage of full frequency reuse over OMA, while increasing the benefit of SIC. The results further suggest that the tolerable \(p_{\rm r}\) can increase under heavier traffic or higher rate requirements, where the rate coverage gain of full frequency reuse is larger.

\begin{figure}[t]
    \centering
    \includegraphics[width=1\linewidth]{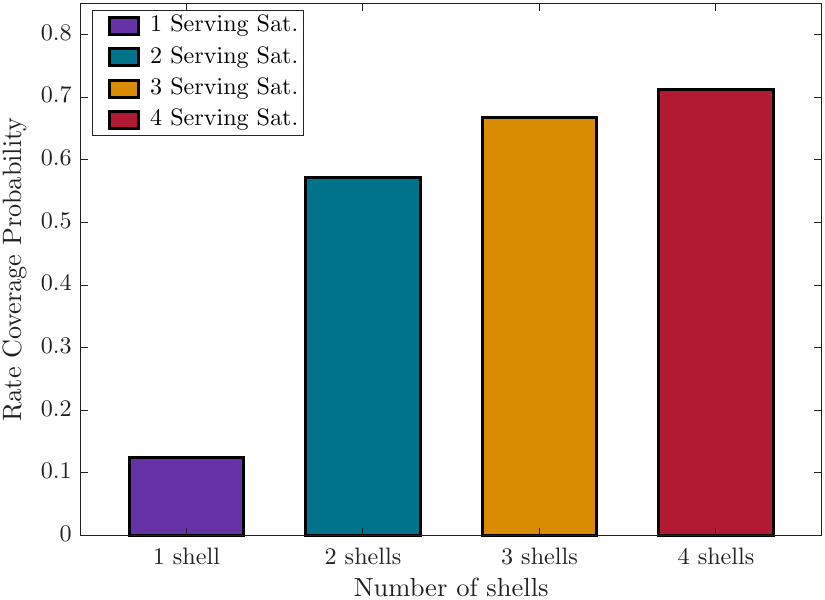}
    \caption{Rate coverage probability versus the number of active shells under a fixed total average satellite count of \(3000\), equally distributed across the active shells. Each occupied shell can provide up to one shell-wise serving opportunity to the footprint.}
    \label{fig:shell increase}
\end{figure}

\subsection{Effect of Additional Serving Shells}

The preceding baseline comparisons allow at most \(N=3\) serving links per footprint in every configuration.
Fig.~\ref{fig:shell increase} instead varies \(N\) and sets the maximum number of shell-wise serving links equal to \(N\), so each additional shell introduces one additional potential serving link.
The \(N=1\) case reduces to the single-shell, single-serving-link configuration.

The figure considers a traffic-hotspot regime with \(m_0=4\), \(K_u=400\), and \(\gamma=20~\mathrm{Mbps}\), while fixing the total average satellite count at \(3000\) and equally distributing it across the active shells.
For \(N=1,2,3,\) and \(4\), the altitude vectors are \([360]\), \([360,480]\), \([360,480,550]\), and \([360,480,550,650]\) km, respectively.
The biases are selected to equalize the shell association probabilities.
Although this is not generally rate-optimal, it provides a simple common load-balancing rule for comparing different numbers of serving shells without separate rate optimization.

The rate coverage probability increases sharply from one to two shells and then more gradually as \(N\) increases.
The additional serving opportunities distribute the traffic load across more links despite the fixed total average satellite count.
The diminishing gain reflects the progressively smaller load-balancing benefit as \(N\) increases, while the longer serving distances of the added upper shells further limit the gain.
This suggests that, in traffic-hotspot regimes, a multi-shell architecture providing multiple shell-wise serving opportunities over the same footprint can achieve higher rate coverage than concentrating the same satellite budget in a single shell.

\textbf{Takeaway 5 (Multi-Shell Serving Opportunities):} In traffic-hotspot regimes, multiple shell-wise serving opportunities over the same footprint can provide higher rate coverage than concentrating a fixed satellite budget in a single shell.

\section{Conclusion}\label{sec:conclusion}

This paper developed a rate coverage framework for multi-shell LEO networks with shell-dependent association biasing and receiver-side SIC.
Multi-altitude geometry tends to favor lower-shell association, while biasing offloads traffic across shells and improves rate coverage through load balancing.
SIC further mitigates the lower-shell interference experienced by users associated with upper shells.
Increasing the satellite count within a fixed set of shells does not necessarily improve per-footprint rate coverage because it does not add shell-wise serving opportunities, although its network-wide benefit can come from supporting more footprints.
In traffic-hotspot regimes, distributing a fixed satellite budget across multiple shells creates additional serving opportunities over the same footprint and can achieve higher rate coverage than concentrating the satellites in a single shell.
Future work includes distributional load modeling, imperfect SIC, and optimization of shell-wise satellite densities for rate coverage.

\appendices
\section{Proof of Lemma \ref{lem:joint_residual_transform}} \label{proof:joint_residual_transform}
Conditioned on \(E_n\) and \(\mathbf R=\mathbf r\), the associated serving distance is \(R_n=r_n<\infty\). For each non-associated shell \(m\neq n\), either \(r_m<\infty\), so that its serving satellite is fixed at distance \(r_m\), or \(r_m=\infty\), so that the shell has no visible satellite and contributes no interference.

For an occupied shell with \(r_m<\infty\), the independent-increments property of the SPPP implies that the remaining visible satellites form an independent SPPP outside \(r_m\). Hence, the aggregate non-serving interference from that shell contributes \(\Psi_m^{\rm non}(s,r_m)\). For the associated shell \(n\), the serving satellite is the desired transmitter, so only the non-serving satellites outside \(r_n\) contribute \(\Psi_n^{\rm non}(s,r_n)\).

After the serving signals from shells \(1,\ldots,i\) have been canceled, a non-associated serving signal remains in the residual interference only when \(m>i\), \(m\neq n\), and \(r_m<\infty\). Averaging over its fading and receive-side antenna-gain state gives \(M_m^{\rm serv}(s,r_m)\). All other cases contribute no serving-interference factor. Independence across shells, fading gains, and antenna-gain variables then yields the product form in \eqref{eq:joint_residual_LT}. Since \(Z_{n,i}=\sigma^2+\mathcal I_n^{(i)}\) with deterministic \(\sigma^2\), multiplying by \(e^{-s\sigma^2}\) gives \eqref{eq:joint_noise_inclusive_LT}.

\section{Proof of Lemma \ref{lem:conditional_residual_cdf}}\label{proof:conditional_residual_cdf}
Since \(\mathcal I_n^{(i)}\) is a sum of nonnegative interference powers, it is nonnegative. Hence, \(F_{\mathcal I_n^{(i)}\mid\mathbf r}(\xi)=0\) for \(\xi<0\). At \(\xi=0\), the CDF equals the probability that the residual interference is exactly zero. This atom can occur when no uncanceled serving interferer and no non-serving interferer are present under the conditioning. For \(s>0\),
\begin{equation}
e^{-s\mathcal I_n^{(i)}}
\rightarrow
\mathds{1}_{\{\mathcal I_n^{(i)}=0\}}
\quad\text{as }s\rightarrow+\infty .
\end{equation}
Since \(0\le e^{-s\mathcal I_n^{(i)}}\le1\), the dominated convergence theorem gives
\begin{equation}
\lim_{s\to+\infty}
\mathcal L_{\mathcal I_n^{(i)}\mid E_n,\mathbf R}(s\mid\mathbf r)
=
\mathbb P(\mathcal I_n^{(i)}=0\mid E_n,\mathbf R=\mathbf r).
\end{equation}
This proves the second case of \eqref{eq:gil_pelaez_joint_cdf}.

For \(\xi>0\), we adopt the Gil--Pelaez inversion formula \cite{di:commlett:2014}, which recovers a CDF from the characteristic function. For a real-valued random variable \(X\) with characteristic function \(\phi_X(\omega)\), the formula gives, at every continuity point \(x\),
\begin{align}\label{eq:general_gil_pelaez}
F_X(x)
=
\frac{1}{2}
-
\frac{1}{\pi}
\int_0^\infty
\frac{
\operatorname{Im}\!\left[
e^{-\mathrm j\omega x}
\phi_X(\omega)
\right]
}{\omega}
d\omega .
\end{align}
In the present setting, the conditional characteristic function of the residual interference is obtained by evaluating the conditional Laplace transform on the imaginary axis:
\begin{equation}
\phi_{\mathcal I_n^{(i)}\mid\mathbf r}(\omega)
=
\mathbb E\!\left[
e^{\mathrm j\omega\mathcal I_n^{(i)}}
\mid E_n,\mathbf R=\mathbf r
\right]
=
\mathcal L_{\mathcal I_n^{(i)}\mid E_n,\mathbf R}
(-\mathrm j\omega\mid\mathbf r).
\end{equation}
Substituting this characteristic function into \eqref{eq:general_gil_pelaez} yields the third case of \eqref{eq:gil_pelaez_joint_cdf}.

Finally, \(Z_{n,i}\) is the residual interference shifted by the deterministic noise power, i.e., \(Z_{n,i}=\sigma^2+\mathcal I_n^{(i)}\). Therefore,
\begin{equation}
F_{Z_{n,i}\mid\mathbf r}(z)
=
\mathbb P(\mathcal I_n^{(i)}\le z-\sigma^2\mid E_n,\mathbf R=\mathbf r).
\end{equation}
If \(z<\sigma^2\), then \(z-\sigma^2<0\), and the probability is zero because \(\mathcal I_n^{(i)}\ge0\). If \(z\ge\sigma^2\), the probability equals \(F_{\mathcal I_n^{(i)}\mid\mathbf r}(z-\sigma^2)\). This proves \eqref{eq:joint_noise_shifted_cdf}.

\section{Proof of Lemma \ref{lem:joint_interval_representation}}\label{proof:joint_interval_representation}
Condition on \(\mathsf S_n=x\), \(\mathbf Y_i=\mathbf y_i\), and \(\mathbf R=\mathbf r\). Since \(\mathcal D_n^{(i)}=D_n^{(i)}\cap(D_n^{(i-1)})^c\), the desired signal must be decodable after stage \(i\) but not immediately before it, and the shell-\(i\) serving power \(y_i\) is still present at stage \(i-1\). The two conditions \(x/Z_{n,i}>\tau_n\) and \(x/(y_i+Z_{n,i})\le\tau_n\) therefore give
\begin{equation}\label{eq:incremental_desired_interval}
\mathcal D_n^{(i)}=\left\{\frac{x}{\tau_n}-y_i\le Z_{n,i}<\frac{x}{\tau_n}\right\}.
\end{equation}
Next, consider an occupied lower shell \(k\le i\). When its serving signal is decoded, the desired signal, the serving signals from shells \(k+1,\ldots,i\), and the residual term \(Z_{n,i}\) all remain in the denominator, so successful decoding requires \(y_k/(x+Z_{n,i}+\sum_{j=k+1}^{i}y_j)>\tau_k\), or equivalently
\begin{equation}\label{eq:sic_stage_upper_bound}
Z_{n,i}<\frac{y_k}{\tau_k}-x-\sum_{j=k+1}^{i}y_j.
\end{equation}
An empty shell imposes no such constraint because its serving signal is absent and the corresponding step is skipped. Combining \eqref{eq:incremental_desired_interval} with \eqref{eq:sic_stage_upper_bound} over all occupied lower shells yields the endpoints \(\ell_{n,i}\) and \(u_{n,i}\) in \eqref{eq:joint_interval_lower} and \eqref{eq:joint_interval_upper}, and hence \eqref{eq:joint_event_interval_equivalence}. If shell \(i\) is empty, then \(y_i=0\) and \(\ell_{n,i}=x/\tau_n\ge u_{n,i}\), so the interval is empty, which correctly reflects that skipping an empty shell cannot create a new decoding success.

Conditioned on \(E_n\) and \(\mathbf R=\mathbf r\), \(Z_{n,i}\) is independent of \(\mathsf S_n\) and \(\mathbf Y_i\) because these quantities depend on disjoint fading, antenna-gain, and point-process components. Hence \(\mathbb P(\ell\le Z_{n,i}<u)=F^{<}_{Z_{n,i}\mid\mathbf r}(u)-F^{<}_{Z_{n,i}\mid\mathbf r}(\ell)\), where the strict distribution function retains any probability mass at the lower endpoint, which proves \eqref{eq:joint_residual_interval_probability}.

\bibliographystyle{IEEEtran}
\bibliography{SIC}

\end{document}